\documentclass[11pt]{amsart}
\usepackage{amsmath,amssymb,amscd,epsfig,url}
\usepackage{fullpage}
\usepackage{times}
\newtheorem{theorem}{Theorem}[section]
\newtheorem{lemma}[theorem]{Lemma}
\newtheorem{proposition}[theorem]{Proposition}
\newtheorem{corollary}[theorem]{Corollary}

\newtheorem{definition}[theorem]{Definition}
\newtheorem{example}[theorem]{Example}
\newtheorem{remark}[theorem]{Remark}

\def\<{{\langle}} \def\>{{\rangle}}       % Inner product
\begin{document}

\author{Aaron Potechin \\
MIT\\
potechin@mit.edu}

\title[Aaron Potechin]
{Monotone switching networks for directed connectivity are strictly more powerful than certain-knowledge 
switching networks}


\begin{abstract}
\noindent 
L (Logarithmic space) versus NL (Non-deterministic logarithmic space) is one of the great open problems 
in computational complexity theory. In the paper ``Bounds on monotone switching networks for directed connectivity'', 
we separated monotone analogues of L and NL using a model called the switching network model. 
In particular, by considering inputs consisting of just a path and isolated vertices, we proved that any monotone 
switching network solving directed connectivity on $N$ vertices must have size at least $N^{\Omega(\lg(N))}$ 
and this bound is tight.\\ 
If we could show a similar result for general switching networks solving directed connectivity, 
then this would prove that $L \neq NL$. However, proving lower bounds for general switching networks solving directed 
connectivity requires proving stronger lower bounds on monotone switching networks for directed connectivity. To work 
towards this goal, we investigated a different set of inputs which we believed to be hard for monotone switching networks 
to solve and attempted to prove similar lower size bounds. Instead, we found that this set of inputs is actually easy 
for monotone switching networks for directed connectivity to solve, yet if we restrict ourselves to certain-knowledge 
switching networks, which are a simple and intuitive subclass of monotone switching networks for directed connectivity, then 
these inputs are indeed hard to solve.\\ 
In this paper, we give this set of inputs, demonstrate a ``weird'' polynomially-sized monotone switching network 
for directed connectivity which solves this set of inputs, and prove that no polynomially-sized certain-knowledge switching 
network can solve this set of inputs, thus proving that monotone switching networks for directed connectivity are 
strictly more powerful than certain-knowledge switching networks.
\end{abstract}
\maketitle
\thispagestyle{empty}
\noindent \textbf{Acknowledgement:}\\
This material is based on work supported by the National Science Foundation Graduate Research Fellowship 
under Grant No. 0645960.
\newpage
%--------------------------------------------------------------------------%
\section{Introduction}\label{intro}
%--------------------------------------------------------------------------%
\setcounter{page}{1}
\noindent $L$ versus $NL$, the problem of whether non-determinism helps in logarithmic space 
bounded computation, is a longstanding open question in computational complexity. At present, only 
a few results are known. It is known that the problem is equivalent to the question of whether there is 
a log-space algorithm for the \textit{directed connectivity} problem, namely given an $N$ vertex directed 
graph $G$ and pair of vertices $s,t$, find out if there is a directed path from $s$ to $t$ in $G$. 
In 1970, Savitch \cite{savitch} gave an $O(\log^{2}N)$-space deterministic algorithm for 
directed connectivity, thus proving that $NSPACE(g(n)) \subseteq DSPACE((g(n)^2))$ for 
every space constructable function $g$. In 1987 and 1988, Immerman \cite{nlconlone} and 
Szelepcsenyi \cite{nlconltwo} independently gave an $O(\log N)$-space non-deterministic algorithm 
for directed \textit{non-connectivity}, thus proving that $NL = co$-$NL$. For the problem of 
\textit{undirected connectivity} (i.e. where the input graph $G$ is undirected), a probabilistic 
algorithm was shown in 1979 using random walks by Aleliunas, Karp, Lipton, $\rm{Lov\acute{a}sz}$, and 
Rackoff \cite{randomwalk}, and in 2005, Reingold \cite{undirectedgraph} gave a deterministic 
$O(\log N)$-space algorithm for the same problem, showing that undirected connectivity is in $L$. Trifonov 
\cite{trifonov} independently gave an $O(\lg{N}\lg{\lg{N}})$ algorithm for undirected connectivity.\\
In terms of monotone computation, in 1988 Karchmer and Wigderson \cite{karchmerwigderson} showed that any 
monotone circuit solving undirected connectivity has depth at least $\Omega((\lg{N})^2)$, thus 
proving that undirected connectivity is not in monotone-$NC^1$ and separating monotone-$NC^1$ 
and monotone-$NC^2$. In 1997 Raz and McKenzie \cite{razmckenzie} proved that 
monotone-$NC \neq$ monotone-$P$ and for any $i$, monotone-$NC^i \neq$ monotone-$NC^{i+1}$.\\
Potechin \cite{potechin} separated monotone analogues of L and NL using the switching network model, 
described in \cite{razborov}. In particular, Potechin \cite{potechin} proved that any monotone switching network 
solving directed connectivity on $N$ vertices must have size at least $N^{\Omega(\lg(N))}$ and this bound is tight. 
To do this, Potechin \cite{potechin} first proved the result for certain-knowledge switching networks, which are 
a simple and intuitive subclass of monotone switching networks for directed connectivity. Potechin \cite{potechin} then 
proved the result for all monotone switching networks solving directed connectivity using Fourier analysis and a partial reduction 
from monotone switching networks for directed connectivity to certain-knowledge switching networks.\\
However, proving good non-monotone bounds requires proving stronger lower bounds on monotone switching networks for directed 
connectivity. The reason is that Potechin \cite{potechin} obtained the above results by considering inputs consisting of just a path 
and isolated vertices, which are the hardest inputs for monotone algorithms to solve but which are easy for non-monotone algorithms 
to solve. To obtain lower bounds on general switching networks for directed connectivity, we must consider different inputs, and 
a lower size bound on all switching networks for directed connectivity solving these inputs implies the same lower bound on all monotone 
switching networks solving these inputs.\\
In this paper, we consider a set of inputs which we originally thought were hard for monotone switching networks to solve. 
Instead, we show that there is a monotone switching network for directed connectivity of polynomial size which solves these inputs, but any 
certain-knowldge switching network solving these inputs must have super-polynomial size. Thus, monotone switching networks 
for directed connectivity are strictly more powerful than certain-knowledge switching networks.\\
To properly state these results, we must first recall some definitions from Potechin \cite{potechin} and introduce a few new definitions. 
These definitions will be used throughout the paper.
%--------------------------------------------------------------------------%
\subsection{Definitions}\label{definitions}
%--------------------------------------------------------------------------%
\begin{definition}\label{modifiedswitchingdefinition}
A switching network for directed connectivity on a set $V(G)$ of vertices with distinguished vertices 
$s,t$ is a tuple $< G', s', t', {\mu}' >$ where $G'$ is an undirected multi-graph with distinguished 
vertices $s'$,$t'$ and ${\mu}'$ is a labeling function such that each edge $e' \in E(G')$ has a label of the form 
$v_1 \to v_2$ or $\neg(v_1 \to v_2)$ for some vertices $v_1, v_2 \in V(G)$.\\
We say that such a switching network is a switching network for directed connectivity on $N$ vertices, 
where $N = |V(G)|$, and we take its size to be $|V(G')|$. A switching network for directed connectivity is 
monotone if it has no labels of the form $\neg(v_1 \to v_2)$.
\end{definition}
\begin{definition}
We say a switching network $G'$ for directed connectivity on a set of vertices $V(G)$ accepts an input graph $G$ if there is a path $P'$
in $G'$ from $s'$ to $t'$ whose edges are all consistent with the input graph $G$ (i.e. of the form $e$ for some 
edge $e \in E(G)$ or $\neg{e}$ for some $e \notin E(G)$).\\
We say a switching network for directed connectivity is sound if it does not accept any input graphs $G$ on the set of vertices $V(G)$ which 
do not have a path from $s$ to $t$.\\
We say a switching network for directed connectivity is complete if it accepts all input graphs $G$ on the set of vertices $V(G)$ which have 
a path from $s$ to $t$.\\
If $G'$ is a switching network for directed connectivity on a set of vertices $V(G)$, then we say that $G'$ solves directed connectivity 
on $V(G)$ if $G'$ is both complete and sound.
\end{definition}
\begin{figure}\label{exampleone}
\centerline{\includegraphics[height=5cm]{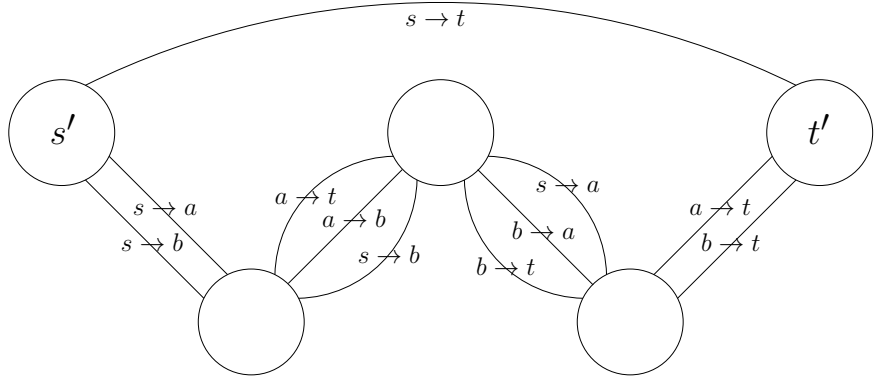}}
\caption{In this figure, we have a monotone switching network that solves directed connectivity on 
$V(G) = \{s,t,a,b\}$. There is a path from $s'$ to $t'$ in $G'$ whose labels are consistent with the input graph $G$ if and only if there 
is a path from $s$ to $t$ in $G$. For example, if we have the edges $s \to a$, $a \to b$, and $b \to t$ 
in $G$, so there is a path from $s$ to $t$ in $G$, then in $G'$, starting from $s'$, we can take the edge 
labeled $s \to a$, then the edge labeled $a \to b$, then the edge labeled $s \to a$, and finally the 
edge labeled $b \to t$, and we will reach $t'$. If in $G$ we have the edges $s \to a$, $a \to b$, 
$b \to a$, and $s \to b$ and no other edges, so there is no path from $s$ to $t$, then in $G'$ there is 
no edge that we can take to $t'$, so there is no path from $s'$ to $t'$.}
\end{figure}
\begin{definition}
Given a nonempty set $I$ of input graphs $G$ on a set of vertices $V(G)$ with distinguished vertices $s,t$, let $I_A$ be the 
set of input graphs in $I$ which contain a path from $s$ to $t$ and let $I_R$ be the set of input graphs in $I$ which 
do not contain a path from $s$ to $t$. If $I_A \neq \emptyset$ and $I_R \neq \emptyset$, we say that a switching network 
$G'$ for directed connectivity on $V(G)$ solves the set of inputs $I$ if $G'$ accepts all $G \in I_A$ and $G'$ does not 
accept any $G \in I_R$. If $I_R = \emptyset$, we say that a switching network $G'$ for directed connectivity on $V(G)$ solves 
the set of inputs $I = I_A$ if $G'$ is sound and $G'$ accepts all $G \in I_A$. If $I_A = \emptyset$, we say that a switching network 
$G'$ for directed connectivity on $V(G)$ solves the set of inputs $I = I_R$ if $G'$ is complete and $G'$ does not accept any $G \in I_R$.
\end{definition}
\begin{figure}\label{forbutnotsolving}
\centerline{\includegraphics[height=5cm]{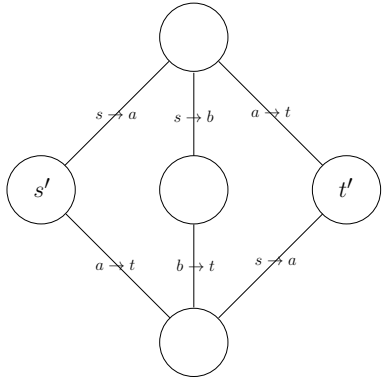}}
\caption{In this figure, we have a monotone switching network $G'$ for directed connectivity on 
$V(G) = \{s,t,a,b\}$. $G'$ accepts an input graph $G$ if and only if $G$ either has the edges $s \to a$ 
and $a \to t$ or has the edges $s \to b$ and $b \to t$ and at least one of the edges $s \to a$, $a \to t$. Thus, $G'$ is sound but 
not complete.}
\end{figure}
\begin{proposition}
If $G'$ is a switching network for directed connectivity on a set of vertices $V(G)$, then $G'$ solves directed connectivity on $V(G)$ if and only if 
$G'$ solves the set $I$ of all possible input graphs $G$ on the set of vertices $V(G)$.
\end{proposition}
\noindent In this paper, we will consider monotone switching networks $G'$ for directed connectivity which solve a set of inputs 
$I = \cup_{i}{\{G_i\}}$ where each input graph $G_i$ contains a path from $s$ to $t$. Thus, in this paper we will only consider sound monotone 
switching networks for directed connectivity, but these switching networks may not be complete.\\
We now define the difficulty of a set of inputs for monotone switching networks for 
directed connectivity.
\begin{definition}
Given a non-empty set of inputs $I$ of input graphs with vertex set $V(G)$, let $M(I)$ be the size of the smallest 
monotone switching network for directed connectivity on $V(G)$ which solves the set of inputs $I$.
\end{definition}
\noindent In this paper, we will consider families of inputs $\mathcal{I} = \{I_n\}$ where 
for each $n$, $I_n$ consists of input graphs on $n$ vertices.
\begin{definition}
We say a family of sets of inputs $\mathcal{I} = \{I_n\}$ is easy for monotone switching networks 
for directed connectivity if there is a polynomial $p(n)$ such that for all $n$, 
$M(I_n) \leq p(n)$. If not, we say that it is hard for monotone switching networks for directed 
connectivity.
\end{definition} 
\noindent Potechin \cite{potechin} introduced a subclass of monotone switching networks for 
directed connectivity called certain-knowledge switching networks which are simple but 
nevertheless have considerable power. They are defined as follows:
\begin{definition}
A knowledge set $K$ is a directed graph with $V(K) = V(G)$, and we represent $K$ by the set of its edges.\\
Given a knowledge set $K$, define the transitive closure $\bar{K}$ of $K$ as follows:\\ 
If there is no path from $s$ to $t$ in $K$, then $\bar{K} = \{v_1 \to v_2: v_1 \neq v_2,$ there is a path from $v_1$ to $v_2$ in $K\}$.\\
If there is a path from $s$ to $t$ in $K$, then $\bar{K}$ is the complete directed graph on $V(G)$.\\ 
Each transitive closure represents an equivalence class of knowledge sets. We say $K_1 = K_2$ if 
$\bar{K_1} = \bar{K_2}$ and we say $K_1 \subseteq K_2$ if $\bar{K_1} \subseteq \bar{K_2}$ as sets.
\end{definition}
\begin{definition}\label{certainknowledgedef}
A certain-knowledge description of a monotone switching network for directed connectivity on a set of vertices 
$V(G)$ is an assignment of a knowledge set $K_{v'}$ to each $v' \in V(G')$. We say a certain-knowledge 
description is valid if the following conditions hold:\\
1. $K_{s'} = \{\}$ and $K_{t'} = \{s \to t\}$.\\
2. If there is an edge $e'$ with label $v_1 \to v_2$ between vertices $v'_1$ and $v'_2$ in $G'$, then \\
$K_{v'_2} \subseteq K_{v'_1} \cup \{v_1 \to v_2\}$ and $K_{v'_1} \subseteq K_{v'_2} \cup \{v_1 \to v_2\}$\\
We say a monotone switching network for directed connectivity is a certain-knowledge switching network 
if there is a valid certain-knowledge description of it.
\end{definition}
\begin{proposition}
All certain-knowledge switching networks for directed connectivity are sound.
\end{proposition}
\begin{figure}\label{examplecertainknowledge}
\centerline{\includegraphics[height=5cm]{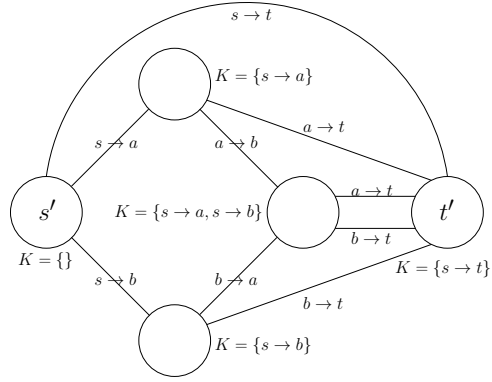}}
\caption{In this figure, we have a certain knowledge switching network $G'$ solving directed connectivity on 
$V(G) = \{s,t,a,b\}$ together with a valid certain-knowledge description for it.}
\end{figure}
\begin{proposition}\label{certainknowledgerules}
The condition that $K_{v'_2} \subseteq K_{v'_1} \cup \{v_1 \to v_2\}$ and 
$K_{v'_1} \subseteq K_{v'_2} \cup \{v_1 \to v_2\}$ 
is equivalent to the condition that we can obtain $K_{v'_2}$ from $K_{v'_1}$ 
using only the following reversible operations on a knowledge set $K$:\\
Operation 1: Add or remove $v_1 \to v_2$.\\
Operation 2: If $v_3 \to v_4,v_4 \to v_5$ are both in $K$ and $v_3 \neq v_5$, add or remove $v_3 \to v_5$.\\
Operation 3: If $s \to t$ is in $K$, add or remove any edge except $s \to t$.\\
If this condition is satisfied, we say we can get from $K_{v'_1}$ to $K_{v'_2}$ with the edge $v_1 \to v_2$.\\
Similarly, two knowledge sets $K_1$ and $K_2$ are equal if and only if we can obtain 
$K_2$ from $K_1$ using only operations 2 and 3.
\end{proposition}
\noindent We define the difficulty of a set of inputs for certain knowledge switching networks in 
a similar way.
\begin{definition}
Given a non-empty set of inputs $I$ of input graphs with vertex set $V(G)$, let $C(I)$ be the size of the smallest 
certain-knowledge switching network for directed connectivity on $V(G)$ which solves the set of inputs $I$.
\end{definition}
\begin{definition}
We say a family of sets of inputs $\mathcal{I} = \{I_n\}$ is easy for certain-knowledge 
switching networks if there is a polynomial $p(n)$ such that for all $n$, $C(I_n) \leq p(n)$. 
If not, then we say it is hard for certain knowledge switching networks.
\end{definition}
\noindent In this paper, we will often consider certain-knowledge switching networks which have a 
valid certain-knowledge description where all knowledge sets contain only edges of the form $s \to v$ for 
some $v \in V(G)$. Accordingly, we introduce the following definitions:
\begin{definition}
Given a set $V \subseteq V(G) \backslash \{s,t\}$, let $K_V$ be the knowledge set 
$\cup_{v \in V}{\{s \to v\}}$
\end{definition}
\begin{definition}\label{basiccertainknowledge}
Given a set of vertices $V(G)$ containing $s,t$, let $G'(V(G),m)$ be the certain-knowledge 
switching network with vertices $t' \cup \{v'_V: V \subseteq V(G), |V| \leq m\}$ and all labeled 
edges allowed by condition 2 of Definition \ref{certainknowledgedef}, where $v'_V$ has knowledge set $K_V$. Note 
that $s' = v'_{\{\}}$.
\end{definition}
\begin{example}
The certain-knowledge switching network shown in Figure 3 is $G'(\{s,t,a,b\},2)$.
\end{example}
\begin{definition}
Given a non-empty set of inputs $I$ of input graphs with vertex set $V(G)$, let $sc(I)$ be the size of the smallest 
$m$ such that $G'(V(G),m)$ solves the set of inputs $I$.
\end{definition}
%--------------------------------------------------------------------------%
\subsection{Our results}
%--------------------------------------------------------------------------%
\noindent We are now ready to properly state our results. We define inputs as follows:
\begin{definition}\label{augmentedinputsdef}
Let $G_0$ be a graph on a vertex set $V(G_0)$ with distinguished vertices $s,t$ and let $V(G)$ be a set of vertices 
which also contains $s$ and $t$.\\
If $\{s,t\},V_0,L,R$ are disjoint subsets of $V(G)$, $V(G) = V_0 \cup L \cup R \cup \{s,t\}$, and $\phi: V_0 \cup \{s,t\} \to V(G_0)$ 
is a one-to-one map with $\phi(s) = s$ and $\phi(t) = t$, then let $G(G_0,V_0,L,R,\phi)$ be the graph with $V(G(G_0,V_0,L,R,\phi)) = V(G)$ and \\
$E(G(G_0,V_0,L,R,\phi)) = \{s \to v: v \in L\} \cup \{v \to t: v \in R\} \cup 
\{v \to w: (\phi(v),\phi(w)) \in E(G_0)\}$
\end{definition}
\begin{figure}\label{exampleinput}
\centerline{\includegraphics[height=10cm]{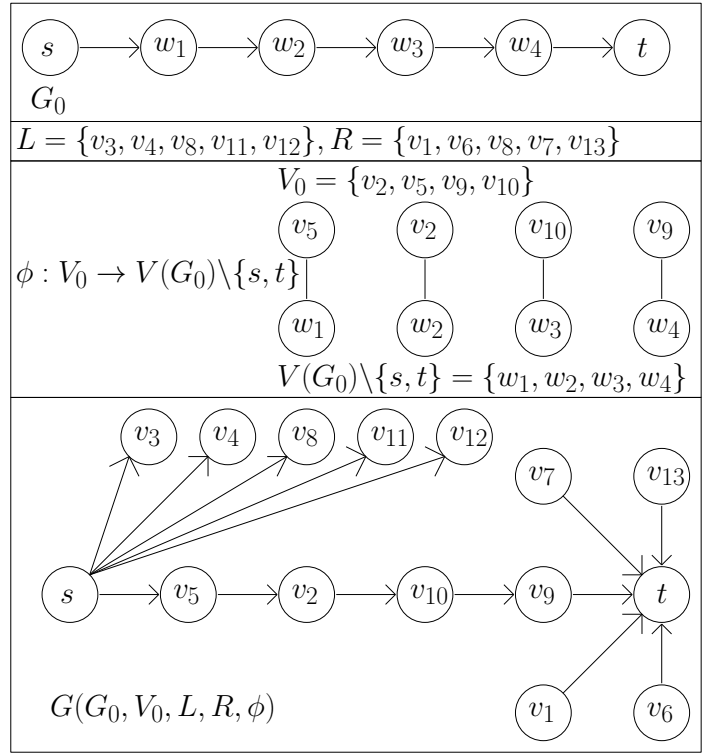}}
\caption{In this figure, we have the input graph of the form $G(G_0,V_0,L,R,\phi)$ for the given $G_0,V_0,L,R,\phi$}
\end{figure}
\noindent In Section \ref{easyinputs}, we prove the following result, which shows that even for certain-knowledge switching networks, 
edges of the form $s \to v, v \in L$ can make the input easier to solve.
\begin{theorem}\label{aneasyinput}
If $G_0 = P$ is a path of length $k+1$ from $s$ to $t$ where $k$ is a constant and if $V(G)$ is a set of vertices of size $N$
containing $s,t$, then letting $I$ be the set of inputs of the form $G(P,V_0,L,\emptyset,\phi)$, $C(I)$ is at most $O(N\lg{N})$ 
(where the constants depend on $k$).
\end{theorem}
\noindent However, the idea used in the proof does not work if the inputs have both edges of the form $s \to v, v \in L$ and edges of the form 
$v \to t, v \in R$. In Section \ref{weirdconstruction}, we use Fourier analysis to describe more sophisticated techniques which 
can use edges of the form $s \to v, v \in L$ and edges of the form $v \to t, v \in R$ with equal effectiveness. Using these techniques, 
we prove the following result:
\begin{theorem}\label{efficientnetwork}
If $I_0 = \cup_{i}{\{G_{0i}\}}$ is a set of input graphs with vertex set $V(G_0)$, all of which contain 
a path from $s$ to $t$, then given a set of vertices $V(G)$ containing $s,t$, letting 
$k = |V(G_0) - 2|$, $N = |V(G)|$, and $m = sc(I_0)$, if we let $I$ be the set of all inputs of the form 
$G(G_{0i},V_0,L,R,\phi)$, then\\
$M(I) \leq 2^{(5m+3)}k^{(3m+2)}{N^3}{\lg{N}}$
\end{theorem}
\noindent Finally, in Section \ref{certainknowledgefails}, we show the following lower bound on certain knowledge switching networks:
\begin{theorem}\label{certainknowledgelowerbound}
If $G_0 = P$ is a path of length $k+1$ from $s$ to $t$, $V(G)$ is a set of vertices of size $N$
containing $s,t$, and $N \geq 10k^2$, then letting $I$ be the set of inputs of the form $G(P,V_0,L,R,\phi)$ and letting \\
$m = 1 + \lfloor{\lg{k}}\rfloor$, $C(I) \geq \frac{1}{2}(\frac{N}{2k(k + \lg{(kN)})})^m$.
\end{theorem}
\noindent We then compare Theorems \ref{efficientnetwork} and \ref{certainknowledgelowerbound} and show the following corollary:
\begin{corollary}\label{separationcorollary}
There is a family of sets of inputs $\mathcal{I}$ such that $\mathcal{I}$ is hard for certain-knowledge switching networks 
but $\mathcal{I}$ is easy for monotone switching networks for directed connectivity.
\end{corollary}
\noindent In other words, monotone switching networks for directed connectivity are strictly more powerful than certain knowledge 
switching networks.
%--------------------------------------------------------------------------%
\subsection{Notation and conventions}
%--------------------------------------------------------------------------%
In this paper, we follow the notation and conventions of Potechin \cite{potechin}. Throughout the paper, we use 
lower case letters (i.e. $a, e, f$) to denote vertices, edges, and functions, and we use upper case letters 
(i.e. $G, V, E$) to denote graphs and sets of vertices and edges. We use unprimed symbols to denote vertices, 
edges, etc. in the directed graph $G$, and we use primed symbols to denote vertices, edges, etc. in the 
switching network $G'$.\\
In this paper, we do not allow graphs to have loops or multiple edges from one vertex to another. 
When a graph has loops or multiple edges from one vertex to another we use the term multi-graph instead. 
We take all paths to be simple (i.e. we do not allow paths to have repeated vertices or edges).\\ 
%--------------------------------------------------------------------------%
\section{An easy set of inputs}\label{easyinputs}
%--------------------------------------------------------------------------%
\noindent It may seem that edges of the form $s \to v$ or $v \to t$ for vertices $v$ which are not on the path from $s$ to $t$ are irrelevant 
and should not make it easier for monotone switching networks to solve the input. However, as we will show, this is not the case.
In this section, we prove Theorem \ref{aneasyinput}, showing that such edges can in fact be useful 
even for certain-knowledge switching networks. We recall the statement of Theorem \ref{aneasyinput} below.
\vskip.1in
\noindent
{\bf Theorem \ref{aneasyinput}.}
{\it
If $G_0 = P$ is a path of length $k+1$ from $s$ to $t$ where $k$ is a constant and if $V(G)$ is a set of vertices of size $N$
containing $s,t$, then letting $I$ be the set of inputs of the form $G(P,V_0,L,\emptyset,\phi)$, $C(I)$ is at most $O(N\lg{N})$ 
(where the constants depend on $k$).
}
\begin{proof}
\noindent The case $k = 0$ is trivial, so we assume that $k \geq 1$.\\
Consider the following procedure for building a certain knowledge switching 
network $G'$:\\
1. Choose an ordering $v_1, \cdots, v_{N-2}$ of the vertices $V(G) \backslash \{s,t\}$\\ 
2. For each $i \in [1,N-2]$, let $V_i = \cup_{j=1}^{i}{v_j}$ and add a vertex with knowledge set 
$K_{V_i}$ to $G'$. We take $V_0 = \{\}$ and $K_{V_0} = \{\}$, so $s'$ is the vertex in $G'$ 
with vertex set $K_{V_0}$.\\
3. Add all edges allowed by condition 2 of Definition \ref{certainknowledgedef} to G'.\\
Clearly, any such $G'$ is sound. We will now show that on average such a $G'$ solves a 
constant fraction of the possible inputs $G(P,V_0,L,\emptyset,\phi)$
\begin{proposition}
Let $w_1, \cdots, w_k$ be the vertices $V(P) \backslash \{s,t\}$. Given an input graph 
of the form $G(P,V_0,L,\emptyset,\phi)$, let $i_j$ be the index such that $\phi(v_{i_j}) = w_j$.\\ 
If $i_1 < i_2 < \cdots < i_k$ then the switching network $G'$ accepts the input 
$G(P,V_0,L,\emptyset,\phi)$.
\end{proposition}
\begin{proof} 
\noindent Note that for any $i \in [1,N-2]$ we can move from 
$K_{V_{i-1}}$ to $K_{V_{i}}$ using the edge $s \to v_i$ unless $i = i_j$ for some $j > 1$. 
However, if $i_{j-1} < i_j$ then $v_{i_{j-1}} \in V_{i-1}$ so we can move from 
$K_{V_{i-1}}$ to $K_{V_{i}}$ using the edge $v_{i_{j-1}} \to v_{i_j}$. We just need to check 
that we can get from $K_{V_{N-2}}$ to $K = \{s \to t\}$ using an edge in 
$G(P,V_0,L,\emptyset,\phi)$. However, this is clear, as $G(P,V_0,L,\emptyset,\phi)$ contains 
an edge $v \to t$ for some $v \in V_{N-2} = V(G) \backslash \{s,t\}$ so we can use the edge 
$v \to t$ to get from $K_{V_{N-2}}$ to $K = \{s \to t\}$.
\end{proof}
\noindent Thus, such a $G'$ accepts any given $G(P,V_0,L,\emptyset,\phi)$ with 
probability at least $\frac{1}{k!}$. Now note that instead of using this construction only once, 
we can use it repeatedly, adding the new vertices to $G'$ each time. Each time we use this 
construction and add the new vertices to $G'$, on average $G'$ will accept at least $\frac{1}{k!}$ 
of the inputs $G(P,V_0,L,\emptyset,\phi)$ that it did not accept before. By the probabilistic method, 
there is always a choice for the ordering $v_1, \cdots, v_{N-2}$ of the vertices 
$V(G) \backslash \{s,t\}$ which will make $G'$ accept at least $\frac{1}{k!}$ 
of the inputs $G(P,V_0,L,\emptyset,\phi)$ that it did not accept before. There are less then $N^k$ distinct 
inputs of the form $G(P,V_0,L,\emptyset,\phi)$, so we can create a certain-knowledge switching network accepting all 
such inputs by repeating this construction at most $1 + \log_{(1 - \frac{1}{k!})^{-1}}{(N^k)} = 1 + 
k\lg{N}\log_{(1 - \frac{1}{k!})^{-1}}{2} \leq 2(k!)k\lg{N}$ times, giving a $G'$ of size at most 
$2N(k!)k\lg{N}$.
\end{proof}
\begin{remark} 
From the proof of Theorem 1.3 of Potechin \cite{potechin}, any certain-knowledge switching network which accepts all paths of length $k+1$ must have 
size at least $\Omega(N^{\lfloor{\lg{k}}\rfloor+1})$. By Corollary 5.22 and Theorem 6.1 of Potechin \cite{potechin}, any sound monotone 
switching network for directed connectivity which accepts all paths of length $k+1$ must have size at least 
$\Omega(N^{\frac{\lfloor{\lg{k}}\rfloor+1}{2}})$. Thus, the edges $\{s \to v: v \in L\}$ make it much easier for both certain-knowledge 
switching networks and monotone switching networks for directed connectivity to solve these inputs.
\end{remark}
%--------------------------------------------------------------------------%
\section{An upper bound for monotone switching networks for directed connectivity}\label{weirdconstruction}
%--------------------------------------------------------------------------%
\noindent We have just shown that certain-knowledge switching networks can effectively use edges of the form $\{s \to v: v \in L\}$. However, 
it seems much harder for a certain-knowledge switching network to use both edges of the form $\{s \to v: v \in L\}$ and 
edges of the form $\{v \to t: v \in R\}$, and we will show in Section \ref{certainknowledgefails} that this is indeed the case. 
In this section, we show that surprisingly, a monotone switching networks for directed connectivity can effectively use both edges of the form 
$\{s \to v: v \in L\}$ and edges of the form $\{v \to t: v \in R\}$.
\noindent We will use the viewpoint of Potechin \cite{potechin} of looking at everything in terms 
of possible cuts of the input graph $G$. Accordingly, we recall the following definitions and facts 
from Potechin \cite{potechin}. 
\begin{definition}
We define an s-t cut (below we use cut for short) of $G$ to be a partition of $V(G)$ into subsets $L(C),R(C)$ such that 
$s \in L(C)$ and $t \in R(C)$. We say an edge $v_1 \to v_2$ crosses $C$ if $v_1 \in L(C)$ and $v_2 \in R(C)$. 
Let $\mathcal{C}$ denote the set of all cuts $C$. $|\mathcal{C}| = 2^{N-2}$, where $N = |V(G)|$.
\end{definition}
\begin{definition}
Given two functions 
$f,g: \mathcal{C} \to \mathbb{R}$, we define $f \cdot g = 2^{2-N}\sum_{C \in \mathcal{C}}{f(C)g(C)}$
\end{definition}
\begin{definition}
Given a set of vertices $V \subseteq V(G)$ that does not include $s$ or $t$, define 
$e_V(C) = {(-1)^{|V \cap L(C)|}}$.
\end{definition}
\begin{proposition}
The set $\{e_V, V \subseteq V(G), s,t \notin V\}$ is an orthonormal basis for the space of functions from 
$\mathcal{C}$ to $\mathbb{R}$. 
\end{proposition}
\noindent Potechin \cite{potechin} took a given monotone switching network for directed 
connectivity and used Fourier analysis to analyze it. Here, we will use suitably defined functions 
from $\mathcal{C}$ to $\mathbb{R}$ to create our switching network.
\begin{definition}\label{switchingnetworkconstruction}
Given a set of functions $H = \{h_{s'}, h_{v'_1}, \cdots, h_{v'_{N'-2}}, h_{t'}\}$ from 
$\mathcal{C}$ to $\mathbb{R}$ where $h_{s'}(C) = -1$ for all cuts $C$ and $h_{t'}(C) = 1$ 
for all cuts $C$, define the switching network $G'(H)$ to have vertices 
$V(G'(H)) = \{s',t',v'_1,\cdots,v'_{N'-2}\}$.\\
For each vertex $v' \in V(G'(H))$ we define $v': \mathcal{C} \to \mathbb{R}$ to be 
$v'(C) = h_{v'}(C)$.\\
For each pair of vertices $v',w'$ in $G'(H)$ and each possible edge $e$ between two vertices 
of $G$, create an edge $e'$ with label $e$ between $v'$ and $w'$ if and only if $v'(C) = w'(C)$ 
for all cuts $C$ which are not crossed by $e$.
\end{definition}
\begin{proposition}
Any monotone switching network for directed connectivity constructed in this way is sound.
\end{proposition}
\begin{proof}
Assume we have a path from $s'$ to $t'$ in $G'(H)$ using only the edges of some input graph 
$G$. Since for each cut $C$, $t'(C) = 1$ and $s'(C) = -1$, there must be some edge 
$e'$ in this path with endpoints $v',w'$ such that $w'(C) \neq v'(C)$. But then by definition, 
if $e$ is the label of $e'$ then $e$ must cross the cut $C$. Thus, for all cuts $C$, $E(G)$ contains 
an edge $e$ crossing $C$ so there must be a path from $s$ to $t$ in $G$.
\end{proof}
\begin{definition}
Let $H$ be a set of functions from $\mathcal{C}$ to $\mathbb{R}$. If $f,g \in H$, 
we say that we can go from $f$ to $g$ with the edge $e$ if $(g - f)(C) = 0$ for 
all cuts $C$ which are not crossed by $e$. We say that we can reach $g$ from $f$ using the set of edges $E$ if there is a 
sequence of functions $h_0, \cdots, h_k$ from $\mathcal{C}$ to $\mathbb{R}$ such that $h_0 = f$, $h_k = g$, and for all $i$ 
$h_i \in H$ and we can get from $h_i$ to $h_{i+1}$ with some edge $e \in E$.
\end{definition}
\begin{proposition}
If $H$ is a set of functions from $\mathcal{C}$ to $\mathbb{R}$ containing the functions $h_{s'} = -e_{\{\}}$ and 
$h_{t'} = -e_{\{\}}$ and if $f,g \in H$, then let $v'_1,v'_2$ be the vertices corresponding to $f,g$ in $G'(H)$. 
We can go from $f$ to $g$ with the edge $e$ if and only if there is an edge in $G'(H)$ with label $e$ between $v'_1$ and $v'_2$. 
Similarly, we can reach $g$ from $f$ using the set of edges $E$ if and only if there is a path from $v'_1$ to $v'_2$ in $G'(H)$ which 
uses only the edges in $E$.
\end{proposition}
\noindent It is useful to know when we can go from one function $f$ to another function $g$ with a given edge $e$. We answer this 
question with the following proposition.
\begin{proposition}\label{keyproposition}
Let $h$ be a function $h: \mathcal{C} \to \mathbb{R}$. Let $v_1,v_2$ be vertices of $G$ which are not $s$ or $t$.\\ 
1. $h(C) = 0$ for all cuts $C$ which cannot be crossed by the edge $s \to v_1$ if and only if $h$ has the form \\
$\sum_{V \subseteq V(G) \backslash \{s,t,v_1\}}{c_V(e_V + e_{(V \cup \{v_1\})})}$\\
2. $h(C) = 0$ for all cuts $C$ which cannot be crossed by the edge $v_1 \to t$ if and only if $h$ has the form \\
$\sum_{V \subseteq V(G) \backslash \{s,t,v_1\}}{c_V(e_V - e_{(V \cup \{v_1\})})}$\\
3. $h(C) = 0$ for all cuts $C$ which cannot be crossed by the edge $v_1 \to v_2$ if and only if $h$ has the form \\
$\sum_{V \subseteq V(G) \backslash \{s,t,v_1,v_2\}}{c_V(e_V - e_{(V \cup \{v_1\})} + e_{(V \cup \{v_2\})} - e_{(V \cup \{v_1,v_2\})})}$\\
\end{proposition}
\begin{proof}
We prove claim 1 as follows. Let $\mathcal{C}_{red}$ be the set of possible cuts $C_{red}$ of $V(G) \backslash \{v_1\}$. Given a function 
$h: \mathcal{C} \to \mathbb{R}$, we define the function $h_{red}: \mathcal{C}_{red} \to \mathbb{R}$ to be 
$h_{red}(C_{red}) = h(C)$ where $L(C) = L(C_{red}) \cup \{v_1\}$ and $R(C) = R(C_{red})$.\\
Writing $h = \sum_{V \subseteq V(G) \backslash \{s,t,v_1\}}{({a_V}{e_V} + {b_V}{e_{(V \cup \{v_1\})}})}$, we have that 
$h_{red} = \sum_{V \subseteq V(G) \backslash \{s,t,v_1\}}{(a_V - b_V)e_V}$\\
$h(C) = 0$ for all cuts which cannot be crossed by the edge $s \to v_1$ if and only if $h_{red} = 0$, which is true if and only if 
$a_V = b_V$ for all $V \subseteq V(G) \backslash \{s,t,v_1\}$, which is true if and only $h$ has the given form, and this completes the 
proof.\\
Claim 2 can be proved in a similar way. To prove claim 3, note that $h(C) = 0$ for all cuts $C$ which cannot be crossed by the 
edge $v_1 \to v_2$ if and only if $h(C) = 0$ for all cuts $C$ which cannot be crossed by the edge $s \to v_2$ and $h(C) = 0$ for all 
cuts $C$ which cannot be crossed by the edge $v_1 \to t$. Using claims 1 and 2, it is easily verified that this is true if and only if 
$h$ has the given form.
\end{proof}
%--------------------------------------------------------------------------%
\subsection{Steps in the Fourier basis}\label{buildingblocks}
%--------------------------------------------------------------------------%
In this subsection, we give examples of what monotone switching networks for directed connectivity can do with the edges 
$s \to v, v \in L$ and $v \to t, v \in R$. We begin with the following simple construction, which illustrates that it is 
relatively easy to use these edges to move around in the Fourier basis.
\begin{proposition}
Let $V = \{v_1, \cdots, v_m\}$ be a non-empty set of vertices with $V \subseteq V(G) \backslash \{s,t\}$ and let 
$V_i = \{v_1, \cdots, v_i\}$. If $H$ is a set of functions containg $h'_{s'} = -e_{\{\}}$ and all of the functions 
$\{\pm{e_{V_i}}, i \in [1,m]\}$ and we have a set of edges $E$ such that for all $i$, $s \to v_i \in E$ or $v_i \to t \in E$, 
then we can reach either $e_{V_m}$ or $-e_{V_m}$ from $-e_{\{\}}$ using the set of edges $E$.
\end{proposition}
\begin{proof}
We prove this result by induction. The base case $m = 1$ is trivial. Assume that we can reach either 
$e_{V_i}$ or $-e_{V_i}$ from $-e_{\{\}}$ using the set of edges $E$. By assumption, $E$ contains either the 
edge $s \to v_{i+1}$ or the edge $v_{i+1} \to t$. If $E$ contains the edge $e = s \to v_{i+1}$ then by Proposition 
\ref{keyproposition} we can go from $e_{V_i}$ to $-e_{V_{i+1}}$ with the edge $e$ and we can 
go from $-e_{V_i}$ to $e_{V_{i+1}}$ with the edge $e$, so the result follows. Similarly, 
if $E$ contains the edge $e = v_{i+1} \to t$ then by Proposition 
\ref{keyproposition} we can go from $e_{V_i}$ to $e_{V_{i+1}}$ with the edge $e$ and we can 
go from $e_{V_i}$ to $e_{V_{i+1}}$ with the edge $e$, so the result follows.
\end{proof}
\noindent We now give several more complicated examples of what monotone switching networks for directed connectivity can do with the edges 
$s \to v, v \in L$ and $v \to t, v \in R$. These examples are motivated by the following idea. Potechin \cite{potechin} associates 
each knowledge set $K$ with the function $K: \mathcal{C} \to \mathbb{R}$ where $K(C) = 1$ if 
there is an edge $e \in K$ crossing $C$ and $0$ otherwise. In particular, \\
$K_V = e_{\{\}} - 2^{(1 - |V|)}\sum_{U \subseteq V}{(-1)^{|U|}e_{U}}$\\
The idea is to mimic these functions with the vertices of $V$ replaced by subsets of vertices.\\
For the rest of this subsection, we will use the following setup:\\
Let $s, t, v_1, \cdots, v_{N-2}$ be the vertices of $V(G)$, let $V_1,V_2,\cdots,V_k$ be disjoint subsets of $V(G) \backslash \{s,t\}$, 
and let $I$ be a non-empty subset of $[1,k]$. Assume that for each $i \in I$ we have a distinguished vertex $v^{*}_i \in V_i$. 
Let $(L,R)$ be a partition of the vertices of $(\cup_{i \in I}{V_i}) \backslash \{s,t\} \backslash (\cup_{i \in I}{\{v^{*}_{i}\}})$ 
and assume that we have a set of functions $H$ from $\mathcal{C}$ to $\mathbb{R}$ 
which contains the functions $h_{s'} = -e_{\{\}}$, $h_{t'} = e_{\{\}}$ and a set of edges $E$ which contains all of the edges 
$s \to v, v \in L$ and $v \to t, v \in R$.
\begin{definition}
Given a subset $V \subseteq V(G)$, define $\theta(V) = (-1)^{(L \cap V)}$. For all $i$, 
define $V_{i0} = \{v^{*}_i\}$ and define $V_{in} = \{v^{*}_i\} \cup (V_i \cap (\cup_{l \leq n}{\{v_l\}}))$
\end{definition}
\begin{lemma}\label{startlemma}
Let $j$ be an element of $I$ and take $I_{red} = I \backslash \{j\}$. If $E$ contains the edge $s \to v^{*}_j$ and $H$ contains the function 
$f = e_{\{\}} - 2^{(1 - |I_{red}|)}\sum_{J \subseteq I_{red}}{(-1)^{|J|}{(\prod_{i \in J}{\theta(V_i)})}e_{(\cup_{i \in J}{V_i})}}$, the function 
$g = e_{\{\}} - 2^{(1 - |I|)}\sum_{J \subseteq I}{(-1)^{|J|}{(\prod_{i \in J}{\theta(V_i)})}e_{(\cup_{i \in J}{V_i})}}$, and
all functions \\
$h_n = e_{\{\}} - 2^{(1 - |I|)}\sum_{J \subseteq I_{red}}{(-1)^{|J|}{(\prod_{i \in J}{\theta(V_i)})}
(e_{(\cup_{i \in J}{V_i})} - \theta(V_{jn})e_{((\cup_{i \in J}{V_i}) \cup V_{jn})})}, n \in [0,N-2]$, \\
then we can reach $g$ from $f$ using the edges of $E$.
\end{lemma}
\begin{proof}
$h_0 - f = 2^{(1 - |I|)}\sum_{J \subseteq I_{red}}{(-1)^{|J|}{(\prod_{i \in J}{\theta(V_i)})}
(e_{(\cup_{i \in J}{V_i})} + e_{((\cup_{i \in J}{V_i}) \cup \{v^{*}_j\})})}$, so by Proposition 
\ref{keyproposition} we can get from $f$ to $h_0$ with the edge $s \to v^{*}_j$.\\
If $n \in [0,N-3]$, then \\
$h_{n+1} - h_n = 2^{(1 - |I|)}\sum_{J \subseteq I_{red}}{(-1)^{|J|}{(\prod_{i \in J}{\theta(V_i)})}
(\theta(V_{j{(n+1)}})e_{((\cup_{i \in J}{V_i}) \cup V_{j{(n+1)}})} - \theta(V_{jn})e_{((\cup_{i \in J}{V_i}) \cup V_{jn})})}$\\
Using Proposition \ref{keyproposition}, it is easily verified that we can get from $h_n$ to $h_{n+1}$ with an edge $e \in E$ where 
$e = s \to v_{n+1}$ if $v_{n+1} \in L \cap V_j$, $e = v_{n+1} \to t$ if $v_{n+1} \in R \cap V_j$, and $e$ is arbitrary 
otherwise (as in this case $h_{n+1} = h_n$).\\
Now note that $g = h_{N-2}$, so we can reach $g$ from $f$ using the edges of $E$, as needed.
\end{proof}
\begin{lemma}\label{endlemma}
Let $l$ be an element of $I$ and let $I_{red} = I \backslash \{l\}$. If $E$ contains the edge $v^{*}_l \to t$ and $H$ contains the function 
$g = e_{\{\}} - 2^{(1 - |I|)}\sum_{J \subseteq I}{(-1)^{|J|}{(\prod_{i \in J}{\theta(V_i)})}e_{(\cup_{i \in J}{V_i})}}$ and
all functions \\
$h_n = e_{\{\}} - 2^{(1 - |I|)}\sum_{J \subseteq I_{red}}{(-1)^{|J|}{(\prod_{i \in J}{\theta(V_i)})}
(e_{(\cup_{i \in J}{V_i})} - \theta(V_{ln})e_{((\cup_{i \in J}{V_i}) \cup V_{ln})})}, n \in [0,N-2]$, \\
then we can reach $e_{\{\}}$ from $g$ using the edges of $E$.
\end{lemma}
\begin{proof}
$e_{\{\}} - h_0 = 2^{(1 - |I|)}\sum_{J \subseteq I_{red}}{(-1)^{|J|}{(\prod_{i \in J}{\theta(V_i)})}
(e_{(\cup_{i \in J}{V_i})} - e_{((\cup_{i \in J}{V_i}) \cup \{v^{*}_l\})})}$, so by Proposition 
\ref{keyproposition} we can get from $h_0$ to $e_{\{\}}$ with the edge $v^{*}_l \to t$.\\
The functions $h_n$ are the same as before, so for all $n \in [0,N-3]$ we can get from $h_{n+1}$ to $h_{n}$ with an edge in $E$. 
Now note that $g = h_{N-2}$, so we can reach $e_{\{\}}$ from $g$ using the edges of $E$, as needed.
\end{proof}
\begin{lemma}\label{progresslemma}
Let $j$ be an element of $I$ and take $I_{red} = I \backslash \{j\}$. If $E$ contains the edge $v^{*}_l \to v^{*}_j$ for some $l \in I_{red}$ 
and $H$ contains the function 
$f = e_{\{\}} - 2^{(1 - |I_{red}|)}\sum_{J \subseteq I_{red}}{(-1)^{|J|}{(\prod_{i \in J}{\theta(V_i)})}e_{(\cup_{i \in J}{V_i})}}$, the function 
$g = e_{\{\}} - 2^{(1 - |I|)}\sum_{J \subseteq I}{(-1)^{|J|}{(\prod_{i \in J}{\theta(V_i)})}e_{(\cup_{i \in J}{V_i})}}$, 
all functions \\
$a_n = f - 2^{(1 - |I|)}\sum_{J \subseteq I_{red} \backslash \{l\}}{(-1)^{|J|}{(\prod_{i \in J}{\theta(V_i)})}(b_n)}$, $n \in [0,N-2]$, where \\
$b_n = (e_{(\cup_{i \in J}{V_i})} - \theta(V_{ln})e_{((\cup_{i \in J}{V_i}) \cup V_{ln})} + 
e_{((\cup_{i \in J}{V_i}) \cup \{v^{*}_j\})} - \theta(V_{ln})e_{((\cup_{i \in J}{V_i}) \cup V_{ln} \cup \{v^{*}_j\})})$,\\
and all functions \\
$h_n = e_{\{\}} - 2^{(1 - |I|)}\sum_{J \subseteq I_{red}}{(-1)^{|J|}{(\prod_{i \in J}{\theta(V_i)})}
(e_{(\cup_{i \in J}{V_i})} - \theta(V_{jn})e_{((\cup_{i \in J}{V_i}) \cup V_{jn})})}, n \in [0,N-2]$, \\
then we can reach $g$ from $f$ using the edges of $E$.
\end{lemma}
\begin{proof}
$a_0 - f = 2^{(1 - |I|)}\sum_{J \subseteq I_{red} \backslash \{l\}}{(-1)^{|J|}{(\prod_{i \in J}{\theta(V_i)})}(b_0)}$ where \\
$b_0 = (e_{(\cup_{i \in J}{V_i})} - e_{((\cup_{i \in J}{V_i}) \cup \{v^{*}_{l}\})} + 
e_{((\cup_{i \in J}{V_i}) \cup \{v^{*}_j\})} - e_{((\cup_{i \in J}{V_i}) \cup \{v^{*}_{l}\} \cup \{v^{*}_j\})})$, so by Proposition 
\ref{keyproposition} we can get from $f$ to $a_0$ with the edge $v^{*}_l \to v^{*}_j$.\\
If $n \in [0,N-3]$, then \\
$a_{n+1} - a_n = 2^{(1 - |I|)}\sum_{J \subseteq I_{red}}{(-1)^{|J|}{(\prod_{i \in J}{\theta(V_i)})}(b_{n} - b_{n+1})}$ where \\
$b_{n} - b_{n+1} = (\theta(V_{l{(n+1)}})e_{((\cup_{i \in J}{V_i}) \cup V_{l{(n+1)}})} - \theta(V_{ln})e_{((\cup_{i \in J}{V_i}) \cup V_{ln})} + 
\theta(V_{l{(n+1)}})e_{((\cup_{i \in J}{V_i}) \cup V_{l{(n+1)}} \cup \{v^{*}_j\})} - 
\theta(V_{ln})e_{((\cup_{i \in J}{V_i}) \cup V_{ln} \cup \{v^{*}_j\})})$\\
Using Proposition \ref{keyproposition}, it is easily verified that we can get from $a_n$ to $a_{n+1}$ with an edge $e \in E$ where 
$e = s \to v_{n+1}$ if $v_{n+1} \in L \cap V_j$, $e = v_{n+1} \to t$ if $v_{n+1} \in R \cap V_j$, and $e$ is arbitrary 
otherwise (as in this case $a_{n+1} = a_n$).\\
The functions $h_n$ are the same as before, so for all $n \in [0,N-3]$ we can get from $h_n$ to $h_{n+1}$ with an edge in $E$. 
Now note that $a_{N-2} = h_0$ and $g = h_{N-2}$, so we can reach $g$ from $f$ using the edges of $E$, as needed.
\end{proof}
%--------------------------------------------------------------------------%
\subsection{The construction}\label{theconstruction}
%--------------------------------------------------------------------------%
We are now ready to construct our monotone switching network and prove Theorem \ref{efficientnetwork}. We recall the statement of Theorem 
\ref{efficientnetwork} below.
\vskip.1in
\noindent
{\bf Theorem \ref{efficientnetwork}.}
{\it
If $I_0 = \cup_{i}{\{G_{0i}\}}$ is a set of input graphs with vertex set $V(G_0)$, all of which contain 
a path from $s$ to $t$, then given a set of vertices $V(G)$ containing $s,t$, letting 
$k = |V(G_0) - 2|$, $N = |V(G)|$, and $m = sc(I_0)$, if we let $I$ be the set of all inputs of the form 
$G(G_{0i},V_0,L,R,\phi)$, then\\
$M(I) \leq 2^{(5m+3)}k^{(3m+2)}{N^3}{\lg{N}}$
}
\begin{proof}
Let $v_1, \cdots, v_{N-2}$ be the vertices of $V(G) \backslash \{s,t\}$ and let 
$w_1, \cdots, w_k$ be the vertices of $V(G_0) \backslash \{s,t\}$.\\
We will use Lemmas \ref{startlemma}, \ref{endlemma}, and \ref{progresslemma} as follows:\\
Let $Q = \{Q_r\}$ be a set of partitions of $V(G) \backslash \{s,t\}$ into $k$ parts, i.e. each 
$Q_{r}$ is of the form $(V_1,\cdots,V_k)$ where all of the $V_i$ are disjoint and $\cup_{i=1}^{k}{V_i} = V(G) \backslash \{s,t\}$.\\ 
Assume that we have an input graph of the form $G = G(G_{0i},V_0,L,R,\phi)$. Recall that $(V_0,L,R,\{s,t\})$ is a partition of $V(G)$, 
$|V_0| = k$, and $\phi: V_0 \to V(G_0)$ is a one-to-one map. Let $v^{*}_i = \phi^{-1}(w_i)$.\\ 
For a non-empty subset $I$ of $[1,k]$ of size at most $m$, if there is an $r$ such that writing $Q_r = (V_1,\cdots,V_k)$ we have that 
$\forall i \in I, V_i \cap V_0 = \{v^{*}_i\}$, then this gives the setup required to use Lemmas \ref{startlemma}, \ref{endlemma}, and 
\ref{progresslemma}.\\ 
Think of the function
$f = e_{\{\}} - 2^{(1 - |I_{red}|)}\sum_{J \subseteq I_{red}}{(-1)^{|J|}{(\prod_{i \in J}{\theta(V_i)})}e_{(\cup_{i \in J}{V_i})}}$
as though it were the knowledge set $K_{(\cup_{i \in I_{red}}{w_i})}$ and think of the function 
$g = e_{\{\}} - 2^{(1 - |I|)}\sum_{J \subseteq I}{(-1)^{|J|}{(\prod_{i \in J}{\theta(V_i)})}e_{(\cup_{i \in J}{V_i})}}$ as though it was 
the knowledge set $K_{(\cup_{i \in I}{w_i})}$. Lemmas \ref{startlemma}, \ref{endlemma}, and \ref{progresslemma} show us how to mimic 
the certain-knowledge swiching network $G'(V(G_0),m)$ defined in Definition \ref{basiccertainknowledge} using only the edges in $E(G)$.\\
However, as we go from one $I \subseteq [1,k]$ to another, we may need to switch from one partition $Q_r$ to another. We show that this can be 
done with the following lemma:
\begin{lemma}\label{switchinggears}
Let $(U_1,\cdots,U_k)$ and $(V_1,\cdots,V_k)$ be two partitions of $V(G) \backslash \{s,t\}$ and define \\
$W_{in} = (V_i \cap (\cup_{l \leq n}{\{v_l\}})) \cup (U_i \cap (\cup_{l > n}{\{v_l\}}))$.\\ 
If we have a subset $I$ such that $\forall i \in I, V_0 \cap U_i = V_0 \cap V_i = \{v^{*}_{i}\}$, then if $H$ contains the function \\
$g_1 = e_{\{\}} - 2^{(1 - |I|)}\sum_{J \subseteq I}{(-1)^{|J|}{(\prod_{i \in J}{\theta(U_i)})}e_{(\cup_{i \in J}{U_i})}}$, the function \\
$g_2 = e_{\{\}} - 2^{(1 - |I|)}\sum_{J \subseteq I}{(-1)^{|J|}{(\prod_{i \in J}{\theta(V_i)})}e_{(\cup_{i \in J}{V_i})}}$, and all 
functions of the form \\
$h_n = e_{\{\}} - 2^{(1 - |I|)}\sum_{J \subseteq I}{(-1)^{|J|}{(\prod_{i \in J}{\theta(W_{in})})}e_{(\cup_{i \in J}{W_{in}})}}$, $n \in [0,N-3]$, 
then we can reach $g_2$ from $g_1$ using only the edges of $E(G)$.
\end{lemma}
\begin{proof}
For $n \in [0,N-3]$, \\
$h_{n+1} - h_n = 2^{(1 - |I|)}\sum_{J \subseteq I}{(-1)^{|J|}
({(\prod_{i \in J}{\theta(W_{in})})}e_{(\cup_{i \in J}{W_{in}})} - {(\prod_{i \in J}{\theta(W_{i{(n+1)}})})}e_{(\cup_{i \in J}{W_{i{(n+1)}}})})}$\\
Letting $W_{J,n} = \cup_{i \in J}{W_{in}}$, 
$h_{n+1} - h_n = 2^{(1 - |I|)}\sum_{J \subseteq I}{(-1)^{|J|}(\theta(W_{J,n})e_{W_{J,n}} - \theta(W_{J,{(n+1)}})e_{W_{J,{(n+1)}}})}$\\
If $W_{J,{(n+1)}} = W_{J,n}$ for all $J \subseteq I$ then $h_{n+1} = h_n$ so we can trivially get from $h_n$ to $h_{n+1}$. If 
$W_{J,{(n+1)}} \neq W_{J,n}$ for some $J \subseteq I$ then $v_{n+1} \in L \cup R$ and either 
$W_{J,{(n+1)}} = W_{J,n} \cup \{v_{n+1}\}$ or $W_{J,n} = W_{J,{(n+1)}} \cup \{v_{n+1}\}$. Using Proposition \ref{keyproposition}, 
it is easily verified that we can get from $h_n$ to $h_{n+1}$ with the edge $s \to v_{n+1}$ if $v_{n+1} \in L$ and we can get from 
$h_n$ to $h_{n+1}$ with the edge $v_{n+1} \to t$ if $v_{n+1} \in R$.\\
$g_1 = h_0$ and $g_2 = h_{N-2}$, so it follows that we can reach $g_2$ from $g_1$ using only the edges of $E(G)$, as needed.
\end{proof}
\noindent We will now construct our set $H$ of functions. Note that there are at most $k^m$ non-empty subsets $I$ of $[1,k]$ of size at most 
$m$ and recall that $Q = \{Q_r\}$ is a set of partitions of $V(G) \backslash \{s,t\}$ into $k$ parts. We will take all functions required by 
Lemmas \ref{startlemma}, \ref{endlemma}, \ref{progresslemma}, and \ref{switchinggears} for all possible input graphs of the form $G(G_{0i},V_0,L,R,\phi)$.\\
We begin by taking all functions of the form 
$e_{\{\}} - 2^{(1 - |I|)}\sum_{J \subseteq I}{(-1)^{|J|}{(\prod_{i \in J}{\theta(V_i)})}e_{(\cup_{i \in J}{V_i})}}$\\
These functions depend only on the partition $Q_r = (V_1,\cdots,V_k)$, the subset $I \subseteq [1,k]$, and the values $\theta(V_i), i \in I$. 
Thus, we have at most $|Q|(k^m)(2^m)$ such functions.\\
We then take all functions of the form \\
$h_n = e_{\{\}} - 2^{(1 - |I|)}\sum_{J \subseteq I_{red}}{(-1)^{|J|}{(\prod_{i \in J}{\theta(V_i)})}
(e_{(\cup_{i \in J}{V_i})} - \theta(V_{jn})e_{((\cup_{i \in J}{V_i}) \cup V_{jn})})}, n \in [0,N-2]$ \\
as required by Lemma \ref{startlemma}, Lemma \ref{endlemma}, and Lemma \ref{progresslemma}. 
Recall that $V_{jn} = \{v^{*}_j\} \cup (V_j \cap (\cup_{l \leq n}{\{v_l\}}))$, so these functions 
depend only on the partition $Q_r = (V_1,\cdots,V_k)$, the subset $I \subseteq [1,k]$, the values $\theta(V_i), i \in I_{red}$, the vertex 
$v^{*}_{j} \in V_j$, the value $\theta(V_{jn})$, and the value of $n$. Thus, we have at most $|Q|(k^m)(2^m)(N^2)$ such functions.\\
Similarly, we take all functions of the form \\
$a_n = f - 2^{(1 - |I|)}\sum_{J \subseteq I_{red} \backslash \{l\}}{(-1)^{|J|}{(\prod_{i \in J}{\theta(V_i)})}(b_n)}$, $n \in [0,N-2]$, where \\
$f = e_{\{\}} - 2^{(1 - |I_{red}|)}\sum_{J \subseteq I_{red}}{(-1)^{|J|}{(\prod_{i \in J}{\theta(V_i)})}e_{(\cup_{i \in J}{V_i})}}$ and \\
$b_n = (e_{(\cup_{i \in J}{V_i})} - \theta(V_{ln})e_{((\cup_{i \in J}{V_i}) \cup V_{ln})} + 
e_{((\cup_{i \in J}{V_i}) \cup \{v^{*}_j\})} - \theta(V_{ln})e_{((\cup_{i \in J}{V_i}) \cup V_{ln} \cup \{v^{*}_j\})})$\\
as required by Lemma \ref{progresslemma}. These functions 
depend only on the partition $Q_r = (V_1,\cdots,V_k)$, the subset $I \subseteq [1,k]$, the values $\theta(V_i), i \in I_{red}$, the vertices  
$v^{*}_l \in V_l, v^{*}_{j} \in V_j$, the value $\theta(V_{ln})$, and the value of $n$. Thus, we have at most $|Q|(k^m)(2^m)(N^3)$ such functions.\\
Finally, we take all functions of the form 
$e_{\{\}} - 2^{(1 - |I|)}\sum_{J \subseteq I}{(-1)^{|J|}{(\prod_{i \in J}{\theta(W_{in})})}e_{(\cup_{i \in J}{W_{in}})}}$ \\
as required by Lemma \ref{switchinggears}. These functions depend only on the partitions $Q_{r_1} = (U_1,\cdots,U_k)$, $Q_{r_2} = (V_1,\cdots,V_k)$, 
the value of $n$, and the values $\theta(W_{in}), i \in I$. Thus, we have at most $|Q|^2(k^m)(2^m)N$ such functions.\\
In total, we have at most $|Q|^2(k^m)(2^m)N + 2|Q|(k^m)(2^m)(N^3)$ functions.
\begin{lemma}
If for every choice of the vertices $v^{*}_1, \cdots, v^{*}_k \in V(G) \backslash \{s,t\}$ and for all subsets 
$I \subseteq [1,k]$ of size at most $m$ there is an $r$ such that writing $Q_r = (V_1,\cdots,V_k)$ we have that 
$\forall i \in I, V_i \cap V_0 = \{v^{*}_i\}$, then for the set of functions $H$ described above, $G'(H)$ accepts all input graphs of 
the form $G = G(G_{0i},V_0,L,R,\phi)$.
\end{lemma}
\begin{proof}
Let $G = G(G_{0i},V_0,L,R,\phi)$ be an input graph and take $v^{*}_i = \phi^{-1}(w_i)$. Since $I_0 = \cup_{i}{\{G_{0i}\}}$ and $sc(I_0) = m$, 
the certain-knowledge switching network $G'(V(G_0),m)$ described in Definition \ref{basiccertainknowledge} accepts the input graph $G_{0i}$, i.e. 
there is a path $P'$ from $s'$ to $t'$ in $G'(V(G_0),m)$ using only the edges of $G_{0i}$. We will show that we can follow each 
step of the certain-knowledge switching network $G'(V(G_0),m)$ in $G'(H)$.\\
Assume that we are currently at the vertex with knowledge set $K_{\{w_i:i \in I\}}$ in $G'(V(G_0),m)$ and we are at a corresponding vertex in $G'(H)$
with function \\
$e_{\{\}} - 2^{(1 - |I|)}\sum_{J \subseteq I}{(-1)^{|J|}{(\prod_{i \in J}{\theta(V_i)})}e_{(\cup_{i \in J}{V_i})}}$ where 
$(V_1, \cdots, V_k) = Q_r$ for some $r$ and \\ 
$\forall i \in I, V_i \cap V_0 = \{v^{*}_i\}$.\\
The knowledge-set of the next vertex on $P'$ has one of the following three forms:\\
1. $K_{\{w_i:i \in I \cup \{j\}\}}$ where $j \notin I$, $|I \cup \{j\}| \leq m$, and either $E(G_{i0})$ contains the edge 
$w_l \to w_j$ for some $l \in I$ or $E(G_{i0})$ contains the edge $s \to w_j$, which implies that $E(G)$ either contains the edge 
$v^{*}_l \to v^{*}_j$ for some $l \in I$ or $E(G_{i0})$ contains the edge $s \to v^{*}_j$. In this case, by assumption 
there is an $r_2$ such that writing $Q_{r_2} = (U_1, \cdots, U_k)$ we have that $\forall i \in I \cup \{j\}, U_i \cap V_0 = \{v^{*}_i\}$. 
We can use Lemma \ref{switchinggears} to \\
reach $e_{\{\}} - 2^{(1 - |I|)}\sum_{J \subseteq I}{(-1)^{|J|}{(\prod_{i \in J}{\theta(U_i)})}e_{(\cup_{i \in J}{U_i})}}$ from 
$e_{\{\}} - 2^{(1 - |I|)}\sum_{J \subseteq I}{(-1)^{|J|}{(\prod_{i \in J}{\theta(V_i)})}e_{(\cup_{i \in J}{V_i})}}$ and then we can use either 
Lemma \ref{startlemma} or \ref{progresslemma} to reach 
$e_{\{\}} - 2^{(1 - |I|)}\sum_{J \subseteq I \cup \{j\}}{(-1)^{|J|}{(\prod_{i \in J}{\theta(U_i)})}e_{(\cup_{i \in J}{U_i})}}$ from 
$e_{\{\}} - 2^{(1 - |I|)}\sum_{J \subseteq I}{(-1)^{|J|}{(\prod_{i \in J}{\theta(U_i)})}e_{(\cup_{i \in J}{U_i})}}$\\
2. $K_{\{w_i:i \in I \backslash \{j\}\}}$ where $j \in I$, $|I| \leq m$, and either $E(G_{i0})$ contains the edge 
$w_l \to w_j$ for some $l \in I$ or $E(G_{i0})$ contains the edge $s \to w_j$, which implies that $E(G)$ either contains the edge 
$v^{*}_l \to v^{*}_j$ for some $l \in I$ or $E(G_{i0})$ contains the edge $s \to v^{*}_j$. In this case, we can use either 
Lemma \ref{startlemma} or \ref{progresslemma} to reach \\ 
$e_{\{\}} - 2^{(1 - |I|)}\sum_{J \subseteq I \backslash \{j\}}{(-1)^{|J|}{(\prod_{i \in J}{\theta(V_i)})}e_{(\cup_{i \in J}{V_i})}}$ from
$e_{\{\}} - 2^{(1 - |I|)}\sum_{J \subseteq I}{(-1)^{|J|}{(\prod_{i \in J}{\theta(V_i)})}e_{(\cup_{i \in J}{V_i})}}$\\
3. $K_{t'}$. Here $|I| \leq m$ and $E(G_{i0})$ contains the edge $w_l \to t$ for some $l \in I$, which implies that $E(G)$ contains the edge 
$v^{*}_l \to t$ for some $l \in I$. In this case, we can use Lemma 
\ref{endlemma} to reach $e_{\{\}}$ from 
$e_{\{\}} - 2^{(1 - |I|)}\sum_{J \subseteq I}{(-1)^{|J|}{(\prod_{i \in J}{\theta(V_i)})}e_{(\cup_{i \in J}{V_i})}}$.\\
Thus, we can follow each step of the certain-knowledge switching network $G'(V(G_0),m)$ in $G'(H)$ using only the edges of $E(G)$. In 
$G'(V(G_0),m)$, we started at $s'$ and ended at $t'$, so there must be a path from $s'$ to $t'$ in $G'(H)$ using only the edges in $E(G)$.
\end{proof}
\noindent Thus, we just need to ensure that for each non-empty subset $I$ of $[1,k]$ of size at most $m$, there is an $r$ such that writing 
$Q_r = (V_1,\cdots,V_k)$ we have that $\forall i \in I, V_i \cap V_0 = \{v^{*}_i\}$. If we are given $v^{*}_1, \cdots, v^{*}_k$ and 
a subset $I \subseteq [1,k]$ of size at most $m$, then if we pick assign each $v \in V(G) \backslash \{s,t\}$ to $V_i$ with probability 
$\frac{1}{k}$ and look at the resulting partition $(V_1,\cdots,V_k)$, the probability that $\forall i \in I, V_i \cap V_0 = \{v^{*}_i\}$ is \\
${k^{-|I|}}(\frac{k-|I|}{k})^{k-|I|} \geq {k^{-|I|}}(\frac{k-|I|}{k})^{k} \geq {k^{-m}}(\frac{k-m}{k})^{k} = 
{k^{-m}}((\frac{k-m}{k})^{\frac{k}{2m}})^{2m} \geq {k^{-m}}(\frac{1}{2})^{2m} = \frac{1}{(4k)^m}$\\
There are at most $N^k$ choices for $v^{*}_1, \cdots, v^{*}_k$ and $k^m$ choices for $I$, so by the probabilistic method we can choose at most 
$1 + \log_{(1-\frac{1}{(4k)^m})^{-1}}((N^k)(k^m)) \leq (4k)^m\lg{((N^k)(k^m))} \leq 2k(4k)^{m}\lg{N}$ distinct $Q_r$ and guarantee that 
for each non-empty subset $I$ of $[1,k]$ of size at most $m$, there is an $r$ such that writing 
$Q_r = (V_1,\cdots,V_k)$ we have that $\forall i \in I, V_i \cap V_0 = \{v^{*}_i\}$. Plugging $|Q| = 2k(4k)^{m}\lg{N}$ into our expression 
for the number of functions $H$ must contain, we find that \\
$|V(G'(H))| = |H| \leq |Q|^2(k^m)(2^m)N + 2|Q|(k^m)(2^m)(N^3) \leq \newline 
2^{(5m+2)}k^{(3m+2)}N(\lg{N})^2 + 2^{(3m+2)}k^{(2m+1)}{N^3}{\lg{N}} \leq 2^{(5m+3)}k^{(3m+2)}{N^3}{\lg{N}}$
\end{proof}
\begin{remark}
This construction can be improved to obtain an $O({N^2}\lg{N})$ upper bound. The best lower bound we have for monotone switching 
networks solving these inputs is $\Omega{(N^2)}$.
\end{remark}
%--------------------------------------------------------------------------%
\section{Lower bounds for certain knowledge switching networks}\label{certainknowledgefails}
%--------------------------------------------------------------------------%
\noindent In this section, we prove a lower bound on certain-knowledge switching networks and deduce that monotone switching networks 
for directed connectivity are strictly more powerful than certain-knowledge switching networks.
\vskip.1in
\noindent
{\bf Theorem \ref{certainknowledgelowerbound}.}
{\it
If $G_0 = P$ is a path of length $k+1$ from $s$ to $t$, $V(G)$ is a set of vertices of size $N$
containing $s,t$, and $N \geq 10k^2$, then letting $I$ be the set of inputs of the form $G(P,V_0,L,R,\phi)$ and letting \\
$m = 1 + \lfloor{\lg{k}}\rfloor$, $C(I) \geq \frac{1}{2}(\frac{N}{2k(k + \lg{(kN)})})^m$.
}
\begin{proof}
Consider which knowledge sets can be useful for accepting a particular input graph $G(P,V_0,L,R,\phi)$. 
We can ignore operation 3 of Proposition \ref{certainknowledgerules}, as if we are ever in a position to use that operation, we can go 
immediately to $t'$ instead. Given an input graph $G(P,V_0,L,R,\phi)$, if we only use operations 1 and 2 of Proposition 
\ref{certainknowledgerules} then we can only obtain edges of the form $v,w$, $v \in \{s\} \cup V_0$, $w \in \{t\} \cup V_0$, of the 
form $s \to v, v \in L$, or of the form $v \to t, v \in R$.\\
By Lemma 3.17 of Potechin \cite{potechin}, any path in a certain-knowledge switching 
network from $s'$ to $t'$ using only the edges of $G(P,V_0,L,R,\phi)$ must pass through at least one vertex $a'$ such that the union 
of the endpoints of the edges in $K_{a'}$ contains at least $m$ of the vertices in $V_0$.
\begin{definition}
Call a knowledge set $K$ useful for the input graph $G(P,V_0,L,R,\phi)$ if $K$ only contains edges of the form 
$v,w$, $v \in \{s\} \cup V_0$, $w \in \{t\} \cup V_0$, of the form $s \to v, v \in L$, or of the form $v \to t, v \in R$, 
the union of the endpoints of the edges in $K$ contains at least $m$ of the vertices in $V_0$, and $K \neq K_{t'} = \{s \to t\}$
\end{definition}
\begin{proposition}\label{certainknowledgerequirement}
If a certain-knowledge switching network $G'$ accepts all inputs of the form $G(P,V_0,L,R,\phi)$, then for each such 
input graph $G = G(P,V_0,L,R,\phi)$, $G'$ contains at least one vertex $v'$ whose knowledge set $K$ is useful for $G$.
\end{proposition}
\begin{lemma}\label{requirementdifficulty}
For any knowledge set $K$, if we choose a random input graph of the form $G = G(P,V_0,L,R,\phi)$ then the probability that $K$ is 
useful for $G$ is at most $2(\frac{2k(k + \lg{(kN)})}{N})^m$.
\end{lemma}
\begin{proof}
Let $V$ be the union of the endpoints of the edges of $K$. If $|V \backslash \{s,t\}| \geq k + m\lg{N}$, then for any input graph 
$G = G(P,V_0,L,R,\phi)$, $V$ contains at least $m\lg{N}$ vertices in $L \cup R$. $K$ can only be useful for $G$ if 
$\forall v \in L \cap V, s \to v \in K, v \to t \notin K$ and $\forall v \in R \cap V, s \to v \notin K, v \to t \in K$. 
Once $V_0$ has been chosen, each other vertex $v \in V(G) - V_0 - \{s,t\}$ is randomly put into $L$ or $R$, so the probability 
of this occuring is at most $2^{-m\lg{N}} = N^{-m}$.\\
If $|V \backslash \{s,t\}| < k + m\lg{N}$, let $x = |V \backslash \{s,t\}|$. The probability that if we choose a random input graph 
$G = G(P,V_0,L,R,\phi)$ that we will have $|V_0 \cap V| = y$ is \\
$p(y) = \frac{\frac{x!}{y!(x-y)!}\frac{(N - 2 - x)!}{(k-y)!(N+y-x-k-2)!}}{\frac{(N-2)!}{k!(N-2-k)!}} = 
\frac{x!k!}{y!(x-y)!(k-y)!}\frac{(N - 2 - x)!(N - 2 - k)!}{(N-2)!(N-x-k-2)!}\frac{(N-x-k-2)!}{(N+y-x-k-2)!} 
\leq (\frac{xk}{N-x-k-2})^y$\\
$N \geq 10k^2$ and $x < k + m\lg{N} < k + \lg{kN} + 1$, so it is easily verified that \\
$\frac{xk}{N-x-k-2} \leq \frac{2k(k + \lg{(kN)})}{N} \leq \frac{1}{2}$ and 
$\sum_{y \geq m}{p(y)} \leq 2p(m) \leq 2(\frac{2k(k + \lg{(kN)})}{N})^m$.\\
Thus the probability that $K$ will be useful for $G$ is at most $2(\frac{2k(k + \lg{(kN)})}{N})^m$, as needed.
\end{proof}
\noindent Combining Proposition \ref{certainknowledgerequirement} and Lemma \ref{requirementdifficulty}, we immediately see that 
$C(I) \geq \frac{1}{2}(\frac{N}{2k(k + \lg{(kN)})})^m$.
\end{proof}
\noindent
{\bf Corollary \ref{separationcorollary}.}
{\it
There is a family of sets of inputs $\mathcal{I}$ such that $\mathcal{I}$ is hard for certain-knowledge switching networks 
but $\mathcal{I}$ is easy for monotone switching networks for directed connectivity.
}
\begin{proof} If $I_0 = \{P\}$ then it follows from Lemmas 3.12 and 3.17 of Potechin \cite{potechin} that 
$sc(I_0) = \lfloor{\lg(k)}\rfloor + 1$. By Theorem \ref{efficientnetwork}, 
$M(I) \leq 2^{(5(\lg{k} + 1)+3)}k^{(3(\lg{k} + 1)+2)}{N^3}{\lg{N}} 
\leq {2^8}{k^{(3\lg{k} + 10)}}{N^3}{\lg{N}}$\\
Taking $k$ to be $2^{\Theta(\sqrt{\lg{N}})}$, $M(I)$ is polynomial in $N$ but $C(I)$ is superpolynomial in $N$.
\end{proof}
%--------------------------------------------------------------------------%
\section{Conclusion}\label{conclusion}
%--------------------------------------------------------------------------%
\noindent In this paper, we have shown that the type of input described in Definition \ref{augmentedinputsdef} is easy for a monotone 
switching network to solve but difficult for a certain-knowledge switching network to solve. Although this result does not give any direct 
progress towards the goal of first proving stronger lower bounds on monotone switching networks for directed connectivity and then 
extending them to the non-monotone case, we believe that this result and the ideas used to prove it are nevertheless very valuable. 
First of all, a major obstacle in proving lower bounds in complexity theory is the difficulty of ruling out ``weird'' algorithms or circuits. 
This result shows that for monotone switching networks for directed connectivity, this difficulty is necessary, as there are some inputs 
for which the simple and intuitive certain-knowledge switching networks are not optimal. Second, Fourier analysis played a key role in 
constructing a small monotone switching network solving the type of input described in Definition \ref{augmentedinputsdef}. This gives 
further evidence that the Fourier analysis approach introduced in Potechin \cite{potechin} is the most fruitful way to analyze monotone 
switching networks for directed connectivity. Finally, this result gives insight into what monotone switching networks for directed 
connectivity can and cannot do. If we prove stronger lower bounds as well, this result may allow us to find large classes of inputs for 
which we can determine almost exactly how hard a given input is for a monotone switching network for directed connectivity to solve.
\newpage
%--------------------------------------------------------------------------%
% The bibliography
%--------------------------------------------------------------------------%
% Bibliography

\end{document}